\documentclass[letterpaper, 10 pt, conference]{ieeeconf}  

\IEEEoverridecommandlockouts                              

\usepackage{hyperref}
\usepackage{url}
\usepackage{graphics} 
\usepackage{times} 
\usepackage{booktabs}       
\usepackage{amsfonts}       
\usepackage{nicefrac}       
\usepackage{microtype}      
\usepackage{algorithm,multirow,xcolor}
\usepackage{algorithmicx}
\usepackage{algpseudocode}
\usepackage{mathtools}
\usepackage{graphicx}
\usepackage{cases}
\usepackage{array}
\usepackage{color}
\usepackage{float}
\usepackage{amssymb}
\usepackage{cite}
\usepackage{wrapfig}
\usepackage{subfigure}
\usepackage{amsmath}
\usepackage{amsthm, amssymb}
\theoremstyle{definition}

\usepackage{soul}

\usepackage{graphicx} 
\usepackage{subfigure}
\usepackage{epstopdf}
\usepackage{cite}
\usepackage{enumerate}
\usepackage[skip=5pt]{caption}

\newtheorem*{remark}{Remark}

\newtheorem{assumption}{Assumption}
\newtheorem{theorem}{Theorem}
\newtheorem{lemma}{Lemma}

\title{\LARGE \bf
	Formation Control for Moving Target Enclosing via \\ Relative Localization
}

\author{Xueming Liu, Kunda Liu,  Tianjiang Hu, and Qingrui Zhang
\thanks{All authors are with the School of Aeronautics and Astronautics, Sun Yat-sen University, Shenzhen 518107, P.R. China (\{\tt \small liuxm93, liukd\}@mail2.sysu.edu.cn, \{\tt \small hutj3, zhangqr9\}@mail.sysu.edu.cn)}
}

\begin{document}
	
	\maketitle
	\thispagestyle{empty}
	\pagestyle{empty}

 	\begin{abstract}
		In this paper, we investigate the problem of controlling multiple unmanned aerial vehicles (UAVs) to enclose a moving target in a distributed fashion based on a relative distance and self-displacement measurements.  A relative localization technique is developed based on the recursive least square estimation (RLSE) technique with a forgetting factor to estimates both the ``UAV-UAV'' and ``UAV-target'' relative positions.  The formation enclosing motion is planned using a coupled oscillator model, which generates desired motion for UAVs to distribute evenly on a circle. The coupled-oscillator-based motion can also facilitate the exponential convergence of relative localization due to its  persistent excitation nature.  Based on the generation strategy of desired formation pattern and relative localization estimates, a cooperative formation tracking control scheme is proposed, which enables the formation geometric center to asymptotically converge to the moving target. The asymptotic convergence performance is analyzed theoretically for both the relative localization technique and the formation control algorithm. Numerical simulations are provided to show the efficiency of the proposed algorithm. Experiments with three quadrotors tracking one target are conducted to evaluate the proposed target enclosing method in real platforms. 
	\end{abstract}
	
	\section{INTRODUCTION}
        %
Formation control has received extensive research attention in recent years due to both its great potential in various applications and theoretical challenges arising in coordination and control schemes \cite{ohSurveyMultiagentFormation2015,Zhang2017JA}. In particular, enclosing a moving target by a group of UAVs in formation has been investigated in various scenarios, such as tracking radio-tagged animals \cite{8250937}, providing external lighting in filming applications \cite{kratkyAutonomousAerialFilming2021}, or tracking ground vehicles from the air \cite{tangOnboardDetectionTrackingLocalization2020}, \emph{etc}. 
 The enclosing task by formation involves cooperative control of multiple vehicles maintaining a particular shape around the target based on available measurements \cite{lopez-nicolasAdaptiveMultirobotFormation2020}. Efficient target enclosing is challenging due to restricted on-board computation resources,  limited sensing capability, and the prerequiste of distributed coordination techniques, \emph{etc}. 
 
 

	Many existing works assume that the target position is known or available to UAVs \cite{brinon-arranzCooperativeControlDesign2014,kleinControlledCollectiveMotion2006}. Such an assumption can simplify the tracking control design process. However, it would lead to a restrictive framework that is not friendly to some real-world applications, \emph{e.g.}, a GPS-denied environment. To address this issue, UAVs are expected to achieve cooperative control using local measurements from onboard sensors. For example, relative position measurement-based approaches have been studied in \cite{hausmanCooperativeMultirobotControl2015,damesDetectingLocalizingTracking2017}. However, the relative position of the target is often not easy to obtain. Hence, it is preferred to find a tracking control algorithm that depends on less sensor knowledge.  Compared with the relative position measurement, the relative distance measurement is easier to acquire and less demanding in terms of the capacity of transmission and storage. Hence, relative distance measurement is mostly preferred in UAVs whose on-board communication and computation capacities are limited in general \cite{yuBearingOnlyCircumnavigation2019}. For instance, \cite{changCollaborativeTargetTracking2021} designs a tracking controller that requires distance measurements between all UAVs with respect to the target in elliptic coordinates. In \cite{6705614BearingMeasurements}, a bearing-based tracking approach is proposed, where the relative position is obtained by estimation.  However, to the authors' best knowledge, it is still an open issue for the moving target enclosing based on local relative distance measurements.

	\begin{figure}[tbp]
		\centering
		\includegraphics[width=0.95\linewidth]{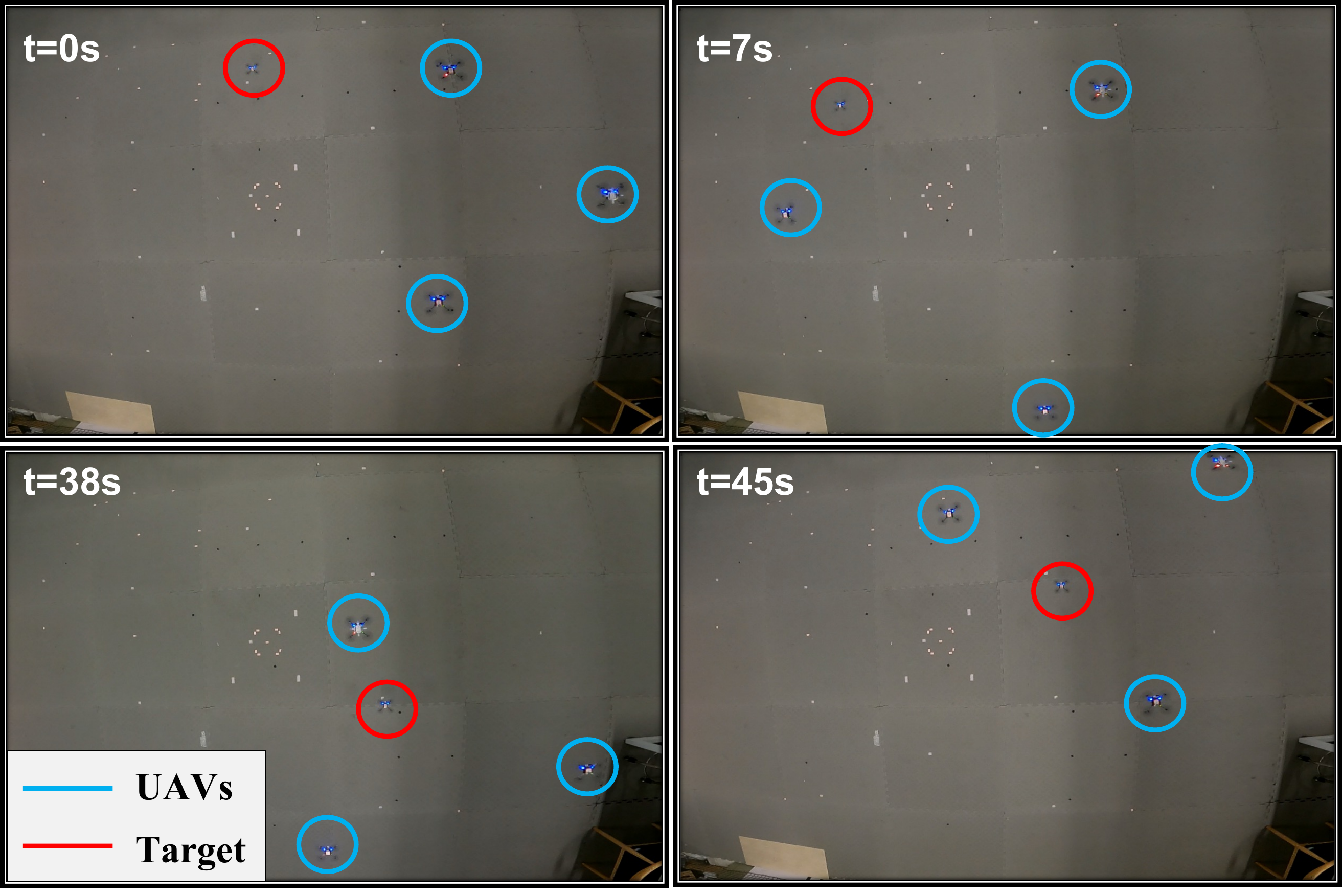}
		\caption{ Experiment of three UAVs (in blue circles) tracking a moving target (in red circle). }
		\label{fig:real}
	\end{figure} 

 


Another important assumption made by most of existing works in target enclosing  is that all UAVs in formation are capable of sensing the target. From a formation control perspective, it implies that all UAVs have access to the leader states. This assumption simplifies the cooperative control efforts among UAVs, leading to impressive tracking performance. However, it is a common scenario that only a few of UAVs can perceive the non-cooperative moving target due to a limited field of view. In some applications, it is also likely that only a few of UAVs are mounted with all the necessary sensors, while the other UAVs are only equipped with cheap or different sensors to reduce the total cost or perform other specific functions. 
The diversity sensing strategy of UAVs will enhance the ability to accomplish complex tasks \cite{ayanianDARTDiversityenhancedAutonomy2019}. In these cases,a distributed cooperative formation control is necessary for all UAVs, which drives all UAVs simultaneously to enclose a target \cite{lopez-nicolasAdaptiveMultirobotFormation2020}.



For efficient target enclosing, it is reasonable for UAVs flying in a time-varying formation \cite{dongTimeVaryingFormationTracking2017}. In time-varying formation flight, the formation patterns or shapes are expected to transform dynamically for dealing with complex tasks, for example, passing through narrow areas by reducing the relative distance \cite{brinon-arranzCooperativeControlDesign2014}, or adjusting the distributed binocular system for better observation of target by affine behaviors \cite{changCollaborativeTargetTracking2021}. More importantly,  time-varying configurations are appropriate to some applications where the number of UAVs in formation could dynamically change due to the loss of some UAVs or addition of new members. Once the number of UAVs increases or decreases, it is necessary to re-coordinate the positions of all UAVs in formation. Hence, a time-varying formation control is necessary for efficient target enclosing \cite{dorflerSynchronizationComplexNetworks2014,sepulchreStabilizationPlanarCollective2007,sepulchreStabilizationPlanarCollective2008}.

In this paper, the moving target enclosing problem is investigated in terms of multiple UAVs based on  distance and displacement measurements. The moving target is expected to be enclosed by a circular time-varying formation and tracked by the geometric center of the formation.  A distributed cooperative formation tracking control algorithm is developed based on relative localization using on-board local measurements. The coupled oscillator model is employed for the design of the cooperative formation control to enclose a moving target in a circular motion. In addition, dynamic formation patterns can be obtained by applying the affine transformation. A relative localization method is proposed based on the recursive least square estimation (RLSE) to estimate both the ``UAV-UAV'' and ``UAV-target'' relative positions. 
A formation tracking controller is eventually designed, which allows the formation geometric center to asymptotically converge to a moving target.
Both numerical simulations and experiments are performed to illustrate the efficiency of the proposed design.
In summary, the overall contributions are fourfold:
\begin{enumerate}
    
    \item A desired enclosing pattern planning module is proposed based on a coupled oscillator-based method in a distributed fashion, which could generate circular motions for UAVs to disperse evenly around a moving target.

    \item A relative localization technique is designed to estimate the relative position of a UAV to its neighbour or the moving target using the relative distance  and self-displacement measurements. 

    
    

    \item A tracking control law that is designed based on the relative position estimate. The proposed algorithm allows the UAVs to enclose the target even if only a part of the UAVs can sense the target.

    \item Theoretical analysis is provided to show the convergence performance of both the relative localization and formation tracking algorithms. 

\end{enumerate}


    \textbf{Notations:}	The following notations are defined below and will be used throughout this paper. $\mathbb{N}$  and $\mathbb{R}$ denote the sets of natural numbers and real numbers, respectively. For $k \in \mathbb{N}$, define $k^+ \triangleq k+1$ and $k^- \triangleq k-1$.  Let $\mathbb{R}^m$ be the $m$-dimensional real vector space. Let $\|\cdot\|$ be the Euclidean norm of a vector. Matrix transpose is denoted by the superscript ``\emph{T}''. $\mathbf{0}_m$ and $\mathbf{1}_m $ represent the $m$-dimensional vector of zeros and ones respectively. $\mathbf{I}$ represents the identity matrix. $ \left\lfloor M\right\rfloor $ denotes the largest integer less than or equal to $M\in\mathbb{R} $.

	\section{Preliminaries} \label{sec:prel}
    In this section, basic concepts and necessary assumptions are introduced for sensing and measurement in our work. The problem description is thereafter presented. 
    
	\subsection{Interaction Topology} \label{subsec:graph}
	The sensing and communication network of $ n $ UAVs is characterized by an undirected graph $ \mathcal{G}=(\mathcal{V},\mathcal{E}) $, where $ \mathcal{V}=\{1,2...n\} $ is the set of vertices, and $ \mathcal{E}\subset\mathcal{V}\times\mathcal{V} $ is the set of edges. A vertex $ i $ represents the $ i^{th} $ UAV. An edge $ e_{ij}=(i,j)\in\mathcal{E} $ with $ i,j\in\mathcal{V},i\ne j$ implies that UAVs $i$ and $j$ are able to communicate with each other. The set of neighbors of a UAV  $i\in\mathcal{V}$ is defined as a set $\mathcal{N}_i=\{j:(i,j)\in\mathcal{E}\}$. The weighted adjacency matrix $\mathcal{A} =[a_{ij}]\in\mathbb{R}^{n\times n}$ of a graph $\mathcal{G}$ describes  the connection strength among UAVs, where $a_{ij}>0$ if and only if $(i,j)\in\mathcal{E}$, and otherwise, $a_{ij}=0$. For an undirected graph, $a_{ij}=a_{ji}$ for all $ e_{ij}\in\mathcal{E} $.

	Additionally, the moving target is labeled as $i=0$, so an extended graph $ \mathcal{\bar{G}}=(\mathcal{\bar{V}},\mathcal{\bar{E}}) $ is introduced with $ \mathcal{\bar{V}} = \{0,1,2,...,n\} $ and $\mathcal{\bar{E}}=\{\bar{e}_{ij}=(i,j)|i,j\in\mathcal{\bar{V}},i\neq j\}$. $\mathcal{\bar{N}}_i=\{j|(i,j)\in\mathcal{\bar{E}}\}$ denotes the neighbor set of a UAV $i$ under $ \mathcal{\bar{G}}$, and  $\mathcal{N}_{0}=\varnothing$. $\mathcal{\bar{A}} =[\bar{a}_{ij}]\in\mathbb{R}^{(n+1)\times (n+1)}$ is the adjacency matrix of $ \mathcal{\bar{G}}$. In this work, $ \mathcal{G}$ should be a subgraph of $ \mathcal{\bar{G}}$. The following assumption is made for the interaction graph $\mathcal{\bar{G}}$.


	\begin{assumption} \label{assump:Topo}
		The graph $\mathcal{G}$ is complete, i.e. each UAV $i\in\mathcal{V}$ communicates with all other UAVs, and let $a_{ij}=1/\mathcal{N}_i$,$i\in\mathcal{V}$. Under the graph $\mathcal{\bar{G}}$, the target $0$ is globally reachable, which means there exist a path from any other vertex $i\in\mathcal{V}$ to $0$, and for any $i\in\mathcal{\bar{V}}$, $\boldsymbol{\Sigma}_{j\in\mathcal{\bar{N}}_i} \bar{a}_{ij}=1$. 
	\end{assumption}

	\subsection{System Dynamics and Measurements}
	Assume that UAVs' motions in the inertial coordinate frame $\mathcal{F}_w $ are modeled by first-order discrete-time dynamics with a bounded velocity in \eqref{eq:dynamics}.
 	\begin{equation}\label{eq:dynamics}
		\left\{
			\begin{aligned}
				& \boldsymbol{p}_{i}(k^{+})=\boldsymbol{p}_{i}(k)+T\boldsymbol{u}_{i}(k) \\
				& \left\lVert \boldsymbol{u}_{i}(k) \right\rVert \leq U_i , \qquad \forall i \in\mathcal{V}  
			\end{aligned}
		\right. 
	\end{equation}	
	where $\boldsymbol{p}_{i}(k) \in \mathbb{R}^{m}$ is the position vector of UAV $i$ at time $kT$ with $T$ denoting a fixed time interval, $k$ is the time step, and $\boldsymbol{u}_{i}(k) \in \mathbb{R}^{m}$ is the velocity control input whose magnitude is bounded by $U_{i}>0$. The target motion is given by
 \[\boldsymbol{p}_{0}(k^{+})=\boldsymbol{p}_{0}(k)+\boldsymbol{v}_{0}(k) \]
 where $\boldsymbol{p}_{0}(k)$ and $ \boldsymbol{v}_{0}(k)$ are the position and self-displacement of the target $0$ in $\mathcal{F}_w $.  



	It is  assumed that the relative distance $d_{ij}(k)$ and relative displacement $\boldsymbol{v}_{ij}(k)$ between $i\in\mathcal{\bar{V}}$ and $j\in\mathcal{\bar{V}}$, with $(i,j)\in\mathcal{\bar{E}} $, can be obtained or measured by certain onboard sensors, so
	\begin{equation}\label{eq:measurements}
		\left\{
			\begin{aligned}
				& d_{ij}(k) = \|\boldsymbol{p}_{i}(k)-\boldsymbol{p}_{j}(k)\| \\
				& \boldsymbol{v}_{ij}(k) = \boldsymbol{v}_{i}(k) - \boldsymbol{v}_{j}(k) \\
				& \boldsymbol{v}_{i}(k) = \boldsymbol{p}_{i}(k^{+}) - \boldsymbol{p}_{i}(k)
			\end{aligned}
		\right.
	\end{equation}
	where $\boldsymbol{v}_{i}(k)$ is the self-displacement of UAV $i$ from time step $k$ to $k^{+}$. For each UAV $i$ and its neighbor $j \in \mathcal{\bar{N}}_{i}$, it is easy to obtain the following formula by the cosine law.
	\begin{equation}\label{eq:cosinelaw}
		\zeta_{ij}(k) \triangleq \frac{1}{2} \left[ d_{ij}^{2}(k^{+})- d_{ij}^{2}(k) -\left\lVert \boldsymbol{v}_{ij}(k) \right\rVert ^{2}  \right]
	\end{equation}	
The following assumption is introduced for the moving target..
	\begin{assumption} \label{assump:sensing}
		The average velocity of the target $ \boldsymbol{u}_{0}(k) \triangleq \boldsymbol{v}_{0}(k) / T$ in time interval $T$ is bounded, \emph{i.e.}, there exist positive constant $U_0$ such that $ \left\lVert \boldsymbol{u}_{0}(k) \right\rVert \leq U_0 $ for all $k$, and $U_0 < U_i, \forall i \in\mathcal{V}$. $ \boldsymbol{u}_{0}(k)$ or $ \boldsymbol{v}_{0}(k) $, is known to all $i\in\mathcal{V}$.
	\end{assumption}

	\begin{remark}
  In practice, $\boldsymbol{v}_{i}(k)$ can be directly obtained from Visual-Inertial Odometry (VIO), Visual Simultaneous Localization and Mapping (V-SLAM), or optical flow by mature hardware modules (\emph{e.g.} \emph{PX4FLOW} and \emph{Flow deck} for \emph{Crazyflies}). In addition, UAV $i$ can measure the distance $d_{ij}(k)$, and receive the displacement $\boldsymbol{v}_{j}(k)$ sent from its neighbor $j\in\mathcal{N}_i$ by Ultra-Wide Band (UWB) sensor or ZigBee. These ranging-based sensors have great potential in many applications, due to the advantages of being lightweight and omnidirectional. Under Assumption \ref{assump:Topo}, there exist at least one UAV $i^*\in\{i:(i,0)\in\mathcal{\bar{E}},i\in\mathcal{V}\}$ can obtain or measure the target's velocity $\boldsymbol{u}_{0}(k)$ and relative distance $d_{i^{*}0}(k)$. In recent works \cite{wangNewMethodDistance2006,zhangExtendingReliabilityMmWave2019,choiMultiTargetTrackingUsing2013,pearceMultiObjectTrackingMmWave2023} , many sensing schemes based on vision, mmWave radar or LiDAR have been proposed to obtain velocity and distance measurements. Hence, UAV $i^*$ can broadcast $\boldsymbol{u}_{0}(k)$ to other through graph $\mathcal{G}$.
	\end{remark}

	\subsection{Problem Statement}
	The whole objective of this paper can be divided into the following sub-tasks.
	\begin{enumerate}
		\item Develop a relative localization method based on the relative distance and relative displacement measurements. Denote the relative position between $(i,j)\in\mathcal{\bar{E}}$ as $\boldsymbol{p}_{ij}(k)$, where $\boldsymbol{p}_{ij}(k)=\boldsymbol{p}_{i}(k)-\boldsymbol{p}_{j}(k)$. In our work, $\boldsymbol{p}_{ij}(k)$ is assumed to be unavailable. Its estimation is specified by $\boldsymbol{\hat{p}}_{ij}(k)$. Hence, the objective is to achieve 
		\begin{equation}\label{eq:estimatorerror}
			\lim_{k \rightarrow \infty } \left\lVert \boldsymbol{\tilde{p}}_{ij}(k) \right\rVert = 0, \qquad \forall (i,j) \in \mathcal{\bar{E}}
		\end{equation}
		where $\boldsymbol{\tilde{p}}_{ij}(k) \triangleq \boldsymbol{\hat{p}}_{ij}(k) - \boldsymbol{p}_{ij}(k)$.
		\item Design a cooperative time-varying formation control strategyto ensure that a group of UAVs reaches a circular motion around a moving non-cooperative target. The center of the formation is defined as 
		\begin{equation}\label{eq:formationcenter}
			\boldsymbol{p}(k) = \frac{1}{n} \sum_{i\in\mathcal{V} } \boldsymbol{p}_{i}(k) 
		\end{equation}
	\end{enumerate}
The whole objective is to allow the projection of formation geometric center on a horizontal planer to converge to the position of the moving target based on the relative localization and distributed cooperative control, namely
		\begin{equation}\label{eq:trackingerror}
			\lim_{k \rightarrow \infty } \left\lVert \boldsymbol{\bar{p}}(k) \right\rVert = 0
		\end{equation}		
		where $\boldsymbol{\bar{p}}(k) \triangleq \boldsymbol{p}(k) - \boldsymbol{p}_0(k) $ is the formation tracking error.
	\section{Circular Formation Design} \label{sec:referformation}
	The position of the target is treated as the center of the desired circular formation motion. In a desired formation pattern, all UAVs should be evenly distributed on the circle with $\boldsymbol{p}_{0}(k)$ as the center and radius $\rho$. Denote $\theta_{i}(k)$ as the angular orientation relative to the $x$-axis on the circle. The desired relative position of UAV $i$ on the circle is formally given as follows
	\begin{equation}\label{eq:desipos}
		\boldsymbol{r}_i(k) = \rho  \left[ \begin{array}{cc} \cos{\theta_i(k)} \\  \sin{\theta_i(k)} \\ \end{array} \right]
	\end{equation}
	The desired relative position is denoted by $\boldsymbol{r}_{ij}(k) \triangleq \boldsymbol{r}_i(k) - \boldsymbol{r}_j(k) $, specially, $\boldsymbol{r}_{i0}(k) = \boldsymbol{r}_i(k) $ with $\boldsymbol{r}_{0}(k) \equiv  \mathbf{0}_2$.  The desired self-displacement is denoted as $\Delta \boldsymbol{r}_{i}(k) \triangleq \boldsymbol{r}_i(k^{+}) - \boldsymbol{r}_i(k)$.
	
	The formation pattern will be achieved through collaboration among UAVs $i\in\mathcal{V}$ in a distributed fashion. To that end, the discrete-time coupled oscillator-based pattern formation method is proposed in this paper.
	\begin{equation}\label{eq:osc}
		\begin{split}
			\theta_i(k^{+}) &= \theta_i(k) + T\omega_{i} \\ &+ T\sum_{j = 1}^{n} \sum_{l = 1}^{n}  \frac{K_{l}a_{ij}}{l}\sin{\left(l\left[\theta_i(k)-\theta_j(k)\right]\right)}
		\end{split}
	\end{equation}
	where $K_{l}\in\mathbb{R} $ are user-defined gains, $\omega_{i}>0$ represents the frequency. 

	\begin{lemma}\label{them:osc}
		Under graph $\mathcal{G}$ and Assumption \ref{assump:Topo}, let $K_{l} > 0 $	for $l\in\{1,...,n-1\} $, and let $K_{n} < 0 $ be sufficiently small. Then a symmetric balanced of regular $n$-sided formation pattern is a locally exponentially stable equilibrium manifold.
	 \end{lemma}
	\begin{proof}
		Equation \eqref{eq:osc} is a discrete form of the Kuramoto model \cite{kleinIntegrationCommunicationControl2007}. In our work, a particular case of $(m, n)$-pattern problem \cite{dorflerSynchronizationComplexNetworks2014,sepulchreStabilizationPlanarCollective2007,sepulchreStabilizationPlanarCollective2008} is considered, where $m=n$.   For further proof, readers of interest can refer to Theorem 7,\cite{sepulchreStabilizationPlanarCollective2008}.
	\end{proof}

	\begin{remark}
		The formation pattern should be dynamically adjusted to effectively deal with complex environments or varied situations during target tracking.  For example, it might  need to pass through a narrow area in a dynamic narrowing manner (see Fig.~\ref{fig:shapedy}). Hence, three affine transformations (translation, scaling, and shearing) are used to achieve this object. Simply define the homogeneous coordinates of a vector $\boldsymbol{q}\in\mathbb{R}^2 $ as $\boldsymbol{q}^{h}=(x,y,1)^{T}\in\mathbb{R}^3 $. The affine transformations, translation, scaling, and shearing are respectively defined by the following matrices
		\begin{equation*}
			\begin{split}
			&\mathbf{T} = \left[ \begin{array}{ccc} 1& 0& T_x \\  0& 1& T_y \\ 0& 0& 1\\ \end{array} \right] , \mathbf{S} = \left[ \begin{array}{ccc} S_x& 0& 0 \\  0& S_y& 0 \\ 0& 0& 1 \\ \end{array} \right] , \\ &\mathbf{H} = \left[ \begin{array}{ccc} 1& H_a& 0 \\  H_b& 1& 0 \\ 0& 0& 1\\ \end{array} \right]
			\end{split}
		\end{equation*}
	where $T_x$, $T_y$, $S_x$, $S_y$, $H_a$, and $H_b$ are external parameters, and $S_x,S_y > 0$. The translation items will make the tracking target out of the formation geometry center and can be expressed in a matrix multiplication of the form $\boldsymbol{q}^{'} = \mathbf{T} \boldsymbol{q}^{h}$ where $\boldsymbol{q}^{'}$ is expressed in homogeneous coordinates. Similarly, scaling and shearing transformation of formation can be reached by $\boldsymbol{q}^{'} = \mathbf{S} \boldsymbol{q}^{h}$, and  $\boldsymbol{q}^{'} = \mathbf{H} \boldsymbol{q}^{h}$ respectively. A scaled circle or ellipse can be obtained by adjusting the parameters.
		Moreover, based on the formation pattern formed by the coupled oscillator, dynamic additions and subtractions of UAVs will be allowed. As shown in Fig.~\ref{fig:shapedy}, UAVs will coordinate their relative positions until an equilibrium is restored. It is very suitable for the failure of a UAV during a task. 
	\end{remark}
	\begin{figure}[tbp]
		\centering
		\includegraphics[width=0.95\linewidth]{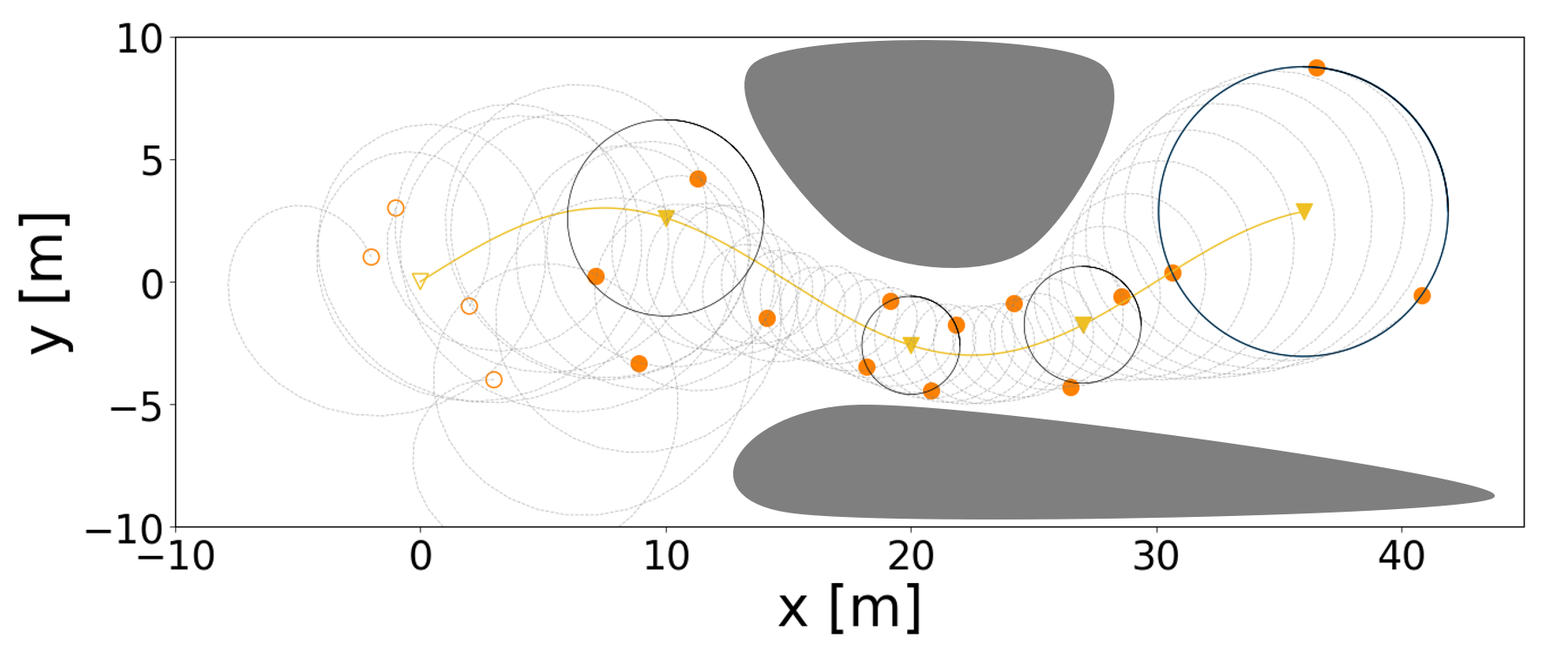}
		\caption{ Simulation of four UAVs (circles) tracking a moving target (triangle) and passing through narrow areas. The initial positions are represented in void style. The radius of time-varying formation is adjusted by affine transformation. At time step $k=250$, UAV $i=4$ is presumed damaged, and the remaining UAVs can still reach a balanced pattern. }
		\label{fig:shapedy}
	\end{figure}

 	\section{Relative localization and tracking control} \label{sec:RL-TC}
	In this section, a relative localization algorithm is proposed to approximate the relative position between two UAVs or UAV and the moving target. A time-varying formation tracking control algorithm is thereafter developed.  
	
	\subsection{Relative Localization Estimator}
	An online localization algorithm is introduced to estimate the relative position between $(i,j)\in\mathcal{\bar{E}}$. The localization algorithm is based on the well-known Recursive Least Square Estimation (RLSE) technique with a forgetting factor. More details on RLSE can be found in \cite{nguyenPersistentlyExcitedAdaptive2020,Zhang2022arXiv}. 
	Based on \eqref{eq:measurements}, the estimate $\boldsymbol{\hat{p}}_{ij}(k)$ is obtained by minimizing the following cost function.
		\begin{equation}\label{eq:costfunc}
			\begin{aligned}
				&J_{ij}(k) = \frac{1}{2} \sum_{l=1}^{k} \beta_f \left[\zeta_{ij}(k)-\boldsymbol{v}_{ij}^{T}(l)\left(\boldsymbol{\hat{p}}_{ij}(k)-\sum_{q=l}^{k}\boldsymbol{v}_{ij}(q)\right)  \right]^{2} \\
				& + \frac{1}{2} \beta_f \left[\boldsymbol{\hat{p}}_{ij}(k)-\sum_{q=1}^{k}\boldsymbol{v}_{ij}(q)-\boldsymbol{\hat{p}}_{ij}(0)\right]^{T} \varGamma_{ij}^{-1}(0) \\
                & \left[\boldsymbol{\hat{p}}_{ij}(k)-\sum_{q=1}^{k}\boldsymbol{v}_{ij}(q)-\boldsymbol{\hat{p}}_{ij}(0)\right]
			\end{aligned}
		\end{equation}		
	where  $\boldsymbol{\hat{p}}_{ij}(0)$ is the initial value of $\boldsymbol{\hat{p}}_{ij}(k)$, $\varGamma_{ij}(k) \in \mathbb{R}^{m\times m}$ with $\varGamma_{ij}(0)$ is a positive definite matrix, and $\beta_f<1$ is a positive constant as forgetting factor.
	By solving $\partial J_{ij}(k)/\partial \boldsymbol{\hat{p}}_{ij}(k)=0$, the least-square estimator is given as follows \cite{nguyenPersistentlyExcitedAdaptive2020}
	\begin{equation}\label{eq:estimator}
		\left\{
			\begin{aligned}
				& \epsilon_{ij}(k) \triangleq \zeta_{ij}(k) - \boldsymbol{v}_{ij}^{T}(k)\boldsymbol{\hat{p}}_{ij}(k) \\
				& \varGamma_{ij}(k^+) = \frac{1}{\beta_f}\left[\varGamma_{ij}(k) - \frac{\varGamma_{ij}(k)\boldsymbol{v}_{ij}(k)\boldsymbol{v}_{ij}^{T}(k)\varGamma_{ij}(k)}{\beta_f+\boldsymbol{v}_{ij}^{T}(k)\varGamma_{ij}(k)\boldsymbol{v}_{ij}(k)}\right] \\
				& \boldsymbol{\hat{p}}_{ij}(k^+) = \boldsymbol{\hat{p}}_{ij}(k) + \boldsymbol{v}_{ij}(k) + \varGamma_{ij}(k^+)\boldsymbol{v}_{ij}(k)\epsilon_{ij}(k)
			\end{aligned}
		\right.
	\end{equation}	

	\begin{remark}
		To ensure the exponential convergence of the relative localization
		estimator, the relative displacement $\boldsymbol{v}_{ij}(k)$ should satisfy the persistent excitation condition. A further discussion of convergence is proposed in Section \ref{sec:convergeanalysis}.
	\end{remark}

	\subsection{Formation Tracking Controller}
	Based on the desired circular formation pattern and the relative position estimates, the cooperative tracking control law as follows
	\begin{equation}\label{eq:formationtrackingcontrol}
		\left\{
			\begin{aligned}
			&  \boldsymbol{\bar{u}}_{i}(k) \triangleq -\beta \sum_{j\in\mathcal{\bar{V}}} \bar{a}_{ij}\left({\hat{\boldsymbol{p}}}_{ij}(k)-\boldsymbol{r}_{ij}(k)\right), \beta < \frac{1}{T} \\
			&  \boldsymbol{u}_i(k) = \pi_{\bar{U}_{i}}\left( \boldsymbol{\bar{u}}_{i}(k)\right) + \frac{1}{T} \Delta \boldsymbol{r}_i(k) + \boldsymbol{u}_{0}(k), i\in\mathcal{V}
			\end{aligned}
		\right.
	\end{equation}
	where $\pi_U\left( \boldsymbol{u}\right)\triangleq \boldsymbol{u} U/\max \{U,\| \boldsymbol{u}\|\}$, $\bar{U}_{i}\in(0,U_i)$ is the upper bound of consensus control term for formation. 

	\begin{remark}
		Note that in \eqref{eq:formationtrackingcontrol}, and under Assumption \ref{assump:Topo} that UAV $i^*$ can sensing the target, the term $\boldsymbol{r}_{i^*0}(k) = \boldsymbol{r}_{i^*}(k)$ drives formation center to track target. Moreover, $\hat{\boldsymbol{p}}_{i^*0}(k)$ can be replaced by the result of other relative positioning methods, which depends on the scenario and onboard sensor configuration. The term $\Delta \boldsymbol{r}_i(k)$ in \eqref{eq:formationtrackingcontrol} is defined by a future desired setpoint, thus it can be considered as a feedforward term. In addition, it can be seen that $\Delta \boldsymbol{r}_i(k)$ , which is differentiation, and $\boldsymbol{\bar{u}}_{i}(k)$, which only involves relative terms. Thus, our tracking control scheme is independent of the common coordinate frame $\mathcal{F}_w $.
	\end{remark}

    \section{convergence analysis} \label{sec:convergeanalysis}

	\subsection{Convergence of Relative Localization Error}
	Given Lemma \ref{them:osc}, to simplify our analysis, we assume that after time $k_oT$ the coupled oscillator is in equilibrium, thus $\Delta \theta_{ij}(k) \triangleq \theta_i(k)-\theta_j(k)=2\pi m/n$ are constants, where $k>k_o, i,j\in\mathcal{V}, m=\{1,2,...,n-1\} $. Therefore, the formation is in uniform circular motion relative to the center of the circle with a period $N = \left\lfloor 2\pi / T\omega_i \right\rfloor $. Further, it is easy to see that a series relative position $W_{ij,l} \triangleq \left[\Delta\boldsymbol{r}_{ij}(l), \Delta\boldsymbol{r}_{ij}(l+1),...,\Delta\boldsymbol{r}_{ij}(l+N-1) \right], \forall l>k_o$ distributed in a circle. Considering real physical systems, assume that $\omega_i$ is bounded. Hence, a $\omega_i$ can be easily designed, such that any two vectors in $W_{ij,l}$ are not equal. In other words, the vectors in $W_{ij,l}$ are not col-linear, which is formally stated as follows:

	\begin{assumption} \label{assump:colinear}
		For any $(i,j)\in\mathcal{\bar{E}}$, there exist $\omega_i<\Omega$, $k_o\in\mathbb{N} $ and $k_1,k_2\in\{l,l+1,...,l+N-1\} $, such that $ rank W_{ij}^* = 2, \forall l>k_o $, where $W_{ij}^* \triangleq \left[\Delta\boldsymbol{r}_{ij}(k_1), \Delta\boldsymbol{r}_{ij}(k_2)\right]$, $N = \left\lfloor 2\pi / T\omega_i \right\rfloor $ and $\Omega$ is a positive constant.
	\end{assumption}

	Based on the above assumptions, it can be seen that 
	\begin{equation}\label{eq:bounddeltar}
		\left\lVert \Delta\boldsymbol{r}_{ij}(k) \right\rVert \leq 2\rho\|\sin(T\omega_i/2)\| \leq T\rho\Omega
	\end{equation}
	Further, denote  $\lambda_{min}(W_{ij}^{*})$ and $\lambda_{max}(W_{ij}^{*})$  as the minimum and the maximum singular values of $W_{ij}^{*}$, respectively. Then we can select the sampling time $T$ that is sufficiently small
	\begin{equation}\label{eq:smallT}
		T < \frac{g(W^{*})}{2\sqrt{2}\bar{U}_i}
	\end{equation}
	where $g(W_{ij}^{*}) = \lambda_{max}(W_{ij}^{*}) - \sqrt{\lambda_{max}^2(W_{ij}^{*})-\lambda_{min}^2(W_{ij}^{*})} $,	and $W^{*} = \arg \max_{W_{ij}^{*} \in \mathcal{W}}g(W_{ij}^{*})$ with$\mathcal{W} = \{W_{ij}^{*}, \forall l \geqslant k_{o},  \forall(i,j)\in\mathcal{\bar{E}}\} $. Based on Assumption \ref{assump:colinear} and condition \eqref{eq:smallT}, the following theorem is reached.
	\begin{theorem}\label{them:rl}
		Suppose Assumption \ref{assump:colinear} holds. Under condition \eqref{eq:smallT}, for any $\left(i,j\right)\in\mathcal{\bar{E}}$, there exist $\alpha_{ij,2}\geq\alpha_{ij,1}>0$, for all $l\geq k_o$ such that $\boldsymbol{v}_{ij}(k)$ satisfy the persistent excitation condition as follows
		\begin{equation}\label{eq:PE}
			\alpha_{ij,1}I \leqslant \boldsymbol{\Phi}_{ij}\left(l\right)\triangleq\sum_{k=l}^{l+N-1}{\boldsymbol{v}_{ij}\left(k\right)\boldsymbol{v}_{ij}^{T}\left(k\right)}\leqslant \alpha_{ij,2}I
		\end{equation}	
		As a consequence, the relative position error $\boldsymbol{\tilde{p}}_{ij}(k)$ converge to $0$ exponentially under the least-square estimator \eqref{eq:estimator}.
	\end{theorem}

	\begin{proof}
		It is obvious that  $\boldsymbol{\Phi}_{ij}\left(l\right)\triangleq \sum_{k=l}^{l+N-1}{\boldsymbol{v}_{ij}\left(k\right)\boldsymbol{v}_{ij}^{T}\left(k\right)}$ is a summation of $\boldsymbol{v}_{ij}\left(k\right)\boldsymbol{v}_{ij}^{T}\left(k\right)$ during a period $N$. Given $\left(i,j\right)\in\mathcal{\bar{E}}$, one has $\boldsymbol{\pi}_{ij}\left(k\right)\triangleq \pi_{\bar{U}_{i}}\left(\boldsymbol{\bar{u}}_{i}\left(k\right)\right)-\pi_{\bar{U}_{j}}\left(\boldsymbol{\bar{u}}_{j}\left(k\right)\right)$, and $\boldsymbol{\pi}_{i0}\left(k\right)\triangleq \pi_{\bar{U}_{i}}\left(\boldsymbol{\bar{u}}_{i}\left(k\right)\right)$ if $j=0$. It can be obtain from \eqref{eq:dynamics},\eqref{eq:formationtrackingcontrol} that
		\begin{equation*}	\boldsymbol{v}_{ij}\left(k\right)=\Delta\boldsymbol{r}_{ij}\left(k\right)+T\boldsymbol{\pi}_{ij}(k)
		\end{equation*}	
	Considering that every term of the formation control law is bounded, and recalling Assumption \ref{assump:sensing} it is easy to see
		\begin{equation*}
			\left\lVert \boldsymbol{\pi}_{ij}\left(k\right) \right\rVert \leq 2\bar{U}_i
		\end{equation*}	
		Since $l$ can be regarded as the initial moment of the iteration, it incurs no loss of generality to only show \eqref{eq:PE} for$\boldsymbol{\Phi}_{ij}(1)$. To show $\alpha_{ij,1}I \leqslant \boldsymbol{\Phi}_{ij}\left(1\right)\leqslant \alpha_{ij,2}I$, it is equivalent to show that for any unit vector $\boldsymbol{x}$,  $$ \alpha_{ij,1} \leqslant \boldsymbol{x}^{T}\boldsymbol{\Phi}_{ij}\left(1\right)\boldsymbol{x} \leqslant \alpha_{ij,2} $$
		Denote
		$$ W_{ij}=\left[\Delta\boldsymbol{r}_{ij}(1),\Delta\boldsymbol{r}_{ij}(2),...,\Delta\boldsymbol{r}_{ij}(N) \right] $$
		$$ V_{ij}=\left[\boldsymbol{\pi}_{ij}(1),\boldsymbol{\pi}_{ij}(2),...,\boldsymbol{\pi}_{ij}(N) \right]$$
		Direct computation shows that
		\begin{equation*}\label{eq:innerproduc}
			\boldsymbol{x}^{T} \boldsymbol{\Phi}_{ij}(1) \boldsymbol{x} = \left\lVert W_{ij}^{T} \right\rVert^{2} + 2T\left\langle W_{ij}^{T}\boldsymbol{x}, V_{ij}^{T}\boldsymbol{x} \right\rangle + T^{2} \left\lVert V_{ij}^{T} \right\rVert^{2}
		\end{equation*}	
		By applying Cauchy-Swartz inequality, there is 
		\begin{equation*}\label{eq:Cauchy-Swartz}
			-\left\lVert W_{ij}^{T}\boldsymbol{x} \right\rVert \left\lVert V_{ij}^{T}\boldsymbol{x} \right\rVert \leq \left\langle W_{ij}^{T}\boldsymbol{x}, V_{ij}^{T}\boldsymbol{x} \right\rangle \leq \left\lVert W_{ij}^{T}\boldsymbol{x} \right\rVert \left\lVert V_{ij}^{T}\boldsymbol{x} \right\rVert
		\end{equation*}	
		Below we will separately show both sides of $\alpha_{ij,1} \leqslant \boldsymbol{x}^{T}\boldsymbol{\Phi}_{ij}\left(1\right)\boldsymbol{x}\leqslant \alpha_{ij,2}$.
		\begin{enumerate}
			\item Equivalently, we will show that there exists $\alpha_{ij,2}>0$ such that $\boldsymbol{x}^{T}\boldsymbol{\Phi}_{ij}\left(1\right)\boldsymbol{x} \leqslant \alpha_{ij,2}$ for any unit vector $\boldsymbol{x}$.
			\begin{equation*}
				\begin{aligned}
	&\boldsymbol{x}^{T}\boldsymbol{\Phi}_{ij}\left(1\right)\boldsymbol{x} \\ &\leq \left\lVert W_{ij}^{T}\boldsymbol{x} \right\rVert^{2} + 2T\left\lVert W_{ij}^{T}\boldsymbol{x} \right\rVert \left\lVert V_{ij}^{T}\boldsymbol{x} \right\rVert + T^{2} \left\lVert V_{ij}^{T}\boldsymbol{x} \right\rVert^{2} \\
					& \leq \left( \left\lVert W_{ij}^{T}\boldsymbol{x} \right\rVert + T\left\lVert V_{ij}^{T}x \right\rVert\right)^{2} \\
					& \leq NT^2\left(\boldsymbol{\rho}\Omega + 2\bar{U}_i\right)^{2} \triangleq \alpha_{ij,2} 
				\end{aligned}
			\end{equation*}	
			by recalling \eqref{eq:bounddeltar}. Thus $\boldsymbol{\Phi}_{ij}\left(1\right) \leqslant \alpha_{ij,2}I$.

			\item Equivalently, we will show that there exists $\alpha_{ij,1}>0$ such that $\boldsymbol{x}^{T}\boldsymbol{\Phi}_{ij}\left(1\right)\boldsymbol{x} \geq \alpha_{ij,1}$ for any unit vector $\boldsymbol{x}$. By recalling Assumption \ref{assump:colinear}, with $W_{ij}^{*} \triangleq \left[\Delta\boldsymbol{r}_{ij}(k_{1}),\Delta\boldsymbol{r}_{ij}(k_{2}) \right], k_1,k_2\in{1,2,...,N}$. We have $\boldsymbol{\Phi}_{ij}\left(1\right)\geq \sum_{k=\{k_1,k_2\}} \boldsymbol{v}_{ij}\left(k\right)\boldsymbol{v}_{ij}^{T}\left(k\right)\triangleq \mathbf{S}_{ij}(1) $.
			
		Similar to the preceding, direct computation shows that
		\begin{equation*}
			\begin{aligned}
				& \boldsymbol{x}^{T}\boldsymbol{\Phi}_{ij}\left(1\right)\boldsymbol{x} \geq \boldsymbol{x}^{T}\mathbf{S}_{ij}\left(1\right)\boldsymbol{x} \\
				& \geq \left\lVert {W_{ij}^{*}}^{T}\boldsymbol{x} \right\rVert^{2} - 2T\left\lVert {W_{ij}^{*}}^{T}\boldsymbol{x} \right\rVert \left\lVert {V_{ij}^{*}}^{T}\boldsymbol{x} \right\rVert + T^{2} \left\lVert {V_{ij}^{*}}^{T} \right\rVert^{2} \\
				& \geq \lambda_{min}^{2} - 2T\lambda_{max}^{2} \left\lVert {V_{ij}^{*}}^{T}\boldsymbol{x} \right\rVert + T^{2} \left\lVert {V_{ij}^{*}}^{T} \right\rVert^{2} 
			\end{aligned}
		\end{equation*}	

		where $0<\lambda_{min}\leq \lambda_{max}$ are respectively the smallest and the largest singular values of $W_{ij}^{*}$, and $V_{ij}^{*} \triangleq \left[\boldsymbol{\pi}_{ij}(k_1),\boldsymbol{\pi}_{ij}(k_2)\right]$. Noticed that $\left\lVert {V_{ij}^{*}} \right\rVert \leq  2\sqrt{2} \bar{U}_i$, we can see that if $\alpha_{ij,1} < \lambda_{min}^{2}$, then the above can be achieved with $2\sqrt{2}T \bar{U}_i \leq \lambda_{max} - \sqrt{\lambda_{max}^{2} - \lambda_{min}^{2}}$. Taken $W_{ij}^{*} = W^{*} $ in \eqref{eq:smallT} , we have shown that $\boldsymbol{\Phi}_{ij}\left(1\right) \geq \alpha_{ij,1}I$
			
		\end{enumerate}

		By combing 1) and 2) the persistent excitation condition of $\boldsymbol{v}_{ij}(k)$ can be obtained. After that, to study the convergence of the estimator \eqref{eq:estimator}, the similar process as in the Theorem IV.1 of \cite{nguyenPersistentlyExcitedAdaptive2020}.
	\end{proof}

	\begin{remark}
		We should note that it is still possible for $\boldsymbol{v}_{ij}(k)$ to satisfy the persistent excitation condition even during the time steps $0<k<k_o$ while \eqref{eq:osc} does not converge. This paper only describes the sufficient condition to meet the persistent excitation condition.
	\end{remark}

	\subsection{Convergence of Formation Tracking Error}
	Based on Theorem \ref{them:rl}, we will show the convergence of formation tracking error $\boldsymbol{\bar{p}}(k)$.
	
	\begin{theorem}\label{them:formationtracking}
		Let Assumption \ref{assump:Topo} and \ref{assump:sensing} hold, and select $T$ by the condition \eqref{eq:smallT}, then the center of the time-varying formation $\boldsymbol{p}(k)$ converges to the target's position $\boldsymbol{p}_0(k)$ exponentially fast under the tracking control law \eqref{eq:formationtrackingcontrol}.
	\end{theorem}
	
	\begin{proof}
		To proof $\lim_{k \rightarrow \infty } \left\lVert \boldsymbol{\bar{p}}(k) \right\rVert = 0$ is equal to proof $\lim_{k \rightarrow \infty } \left\lVert \sum_{i\in\mathcal{V} } \left( \boldsymbol{p}_{i}(k) - \boldsymbol{p}_{0}(k) \right) \right\rVert = 0$. Further, it is equal to the following convergence for any $i\in\mathcal{V}$
		\begin{equation}\label{eq:trackerro}
			\lim_{k \rightarrow \infty } \sup \left\lVert \boldsymbol{\bar{p}}_{i}(k) \right\rVert = 0
		\end{equation}
		where $\boldsymbol{\bar{p}}_{i}(k) \triangleq  \boldsymbol{p}_{i}(k) - \boldsymbol{r}_{i}(k) - \boldsymbol{p}_{0}(k)  $. By recalling \eqref{eq:dynamics} and \eqref{eq:formationtrackingcontrol}, we can find that
		\begin{equation}\label{eq:trackerrody}
			\begin{aligned}
			\boldsymbol{\bar{p}}_{i}(k^{+}) &= \boldsymbol{\bar{p}}_{i}(k) + T\pi_{\bar{U}_{i}}\left( \boldsymbol{\bar{u}}_{i}(k)\right)\\
			&= \gamma_{i}(k)\boldsymbol{\bar{p}}_{i}(k) + \sum_{j\in\mathcal{N}_{i}} \gamma_{ij}(k)\boldsymbol{\bar{p}}_{j}(k) - \boldsymbol{e}_i(k)
			\end{aligned}
		\end{equation}
		where $\gamma_{ij}(k) = T\beta s_i(k) \bar{a}_{ij} $, $\gamma_{i}(k)=1-\sum_{j\in\mathcal{\bar{N}}_{i}}\gamma_{ij}(k) = 1-T\beta s_i(k)$, by using $\boldsymbol{\Sigma}_{j\in\mathcal{\bar{N}}_i} \bar{a}_{ij}=1$, and $s_i(k) \triangleq U/\max \{U,\| \boldsymbol{u}\|\} \in (0,1]$, hence $\gamma_{i}(k), \gamma_{ij}(k)\in(0,1)$. As the result of Theorem \ref{them:rl}, there exist $\delta >0$ and $ 0< \lambda <1 $ such that $\left\lVert \boldsymbol{e}_i(k) \right\rVert \leq \delta \lambda^{k} $ for all $i\in\mathcal{V}$. By applying the trangle inequality to \eqref{eq:trackerrody}, we get 
		\begin{equation}\label{eq:trackerrodytria}
			\begin{aligned}
			\left\lVert \boldsymbol{\bar{p}}_{i}(k^{+})\right\rVert 
			&\leq  \gamma_{i}(k)\left\lVert \boldsymbol{\bar{p}}_{i}(k)\right\rVert  + \sum_{j\in\mathcal{N}_{i}} \gamma_{ij}(k)\left\lVert \boldsymbol{\bar{p}}_{j}(k)\right\rVert  + \delta \lambda^{k}
			\end{aligned}
		\end{equation}
		Accordingly, we can rewrite the inequations in matrix form for all $i\in\mathcal{V}$. Thus, we get the following componentwise inequality
		\begin{equation}\label{eq:trackerrodytriamatrix}
			\mathcal{P} (k+1) \leq \Lambda(k)\mathcal{P}(k) + \mathbf{1}_n \delta \lambda^{k}
		\end{equation}
		where $\mathcal{P}(k)\triangleq \left[ \left\lVert \boldsymbol{\bar{p}}_{1}(k)\right\rVert ,...,\left\lVert \boldsymbol{\bar{p}}_{n}(k)\right\rVert  \right] ^{T} $, and $$ \Lambda(k) \triangleq \left[ \begin{array}{cccc} \gamma_{1}(k)& \gamma_{12}(k)& \cdots&  \gamma_{1n}(k) \\  \gamma_{21}(k)& \gamma_{2}(k)& \cdots&  \gamma_{2n}(k) \\ \vdots & \vdots & \ddots & \vdots \\ \gamma_{n1}(k)& \gamma_{n2}(k)& \cdots&  \gamma_{n}(k) \\ \end{array} \right]  $$
		Clearly, $\Lambda(k)$ is a nonnegative matrix with each row sum $ \Lambda_{i}(k)\mathbf{1}_n \triangleq \gamma_{i}(k) + \sum_{j\in\mathcal{N}_{i}} \gamma_{ij}(k) = 1 - \gamma_{i0}(k) \leqslant 1 $. The equality occurs if and only if $0\notin \mathcal{\bar{N}}_{i} $. Thus, $\Lambda(k)\mathbf{1}_n\leqslant \mathbf{1}_n$, and we get $$ \mathcal{P} (k) \leq \mathcal{P}(k-1) + \mathbf{1}_n \delta \lambda^{k-1} \leq \mathcal{P}(0) + \mathbf{1}_n \delta \sum_{t=0}^{k-1}\lambda^{l}$$
		which means $ \boldsymbol{\bar{p}}_{i}(k) $ is bounded.
		After showing the boundedness of $ \boldsymbol{\bar{p}}_{i}(k) $, we will establish the convergence based on \eqref{eq:trackerrodytriamatrix} in the following.

		Under Assumption \ref{assump:Topo}, we can classify the UAVs into two groups $ \mathsf{O}_1 $ and $ \mathsf{O}_2 $. The shortest path from $i\in \mathsf{O}_1$ to the target $0$ be of length $l=1$, and $l=2$ for $i\in \mathsf{O}_2$. Denote $\mathsf{Q}(k) = \max_{i\in\mathcal{V}} \left\lVert \boldsymbol{\bar{p}}_{i}(k)\right\rVert  $. We will show that for any $i\in\mathsf{O}_l, l =1,2$, there exist $ \eta_{i}\in(0,1) $ such that
		\begin{equation}\label{eq:pibar}
			\left\lVert \boldsymbol{\bar{p}}_{i}(k+2)\right\rVert \leq \eta_{i} \mathsf{Q}(k) + \delta \sum_{t=k}^{k+1}\lambda^{t}
		\end{equation}

		To show \eqref{eq:pibar}, we note that \eqref{eq:trackerrodytriamatrix} implies that $$ \mathcal{P} (k+2) \leq \Lambda(k+1)\Lambda(k)\mathcal{P}(k) + \mathbf{1}_n \delta \sum_{t=k}^{k+1}\lambda^{t} $$

		For $i\in \mathsf{O}_1$,  it is clear that $ \Lambda_{i}(k)\mathbf{1}_n = 1 - \gamma_{i0}(k) \triangleq  \eta_{i}(k) < 1 $, where $ \gamma_{i0}(k)>0 $ due to $ 0\in \mathcal{\bar{N}}_{i}$. Note that $ \Lambda(k)\mathbf{1}_n < \kappa_{i}(k) $, where $ \kappa_{i}(k) \in \mathbb{R}^n $ is obtained by replacing the $i^{th}$ entry of $\mathbf{1}_n$ with $\eta_{i}(k)$. Then, we have
		\begin{equation*}
			\begin{split}
				\Lambda_i(k+1)\Lambda(k)\mathbf{1}_n &\leq \Lambda_i(k+1)\kappa_{i}(k) \\
				&\leq \gamma_{i}(k+1)\eta_{i}(k) + \sum_{j\in\mathcal{N}_{i}} \gamma_{ij}(k+1) \\
				&= \eta_{i}(k) - \gamma_{i}(k+1)\gamma_{i0}(k+1) \\
				&\triangleq \eta_{i} \leqslant \eta_{i}(k) < 1
			\end{split}
		\end{equation*}
		
		For $j\in \mathsf{O}_2$, which is connected to the above $i\in \mathsf{O}_1$, it follows from $ \Lambda(k)\mathbf{1}_n < \kappa_{i}(k) $ that
		\begin{equation*}
			\begin{split}
				\Lambda_j(k+1)&\Lambda(k)\mathbf{1}_n \leq \Lambda_j(k+1)\kappa_{i}(k) \\
				& = \Lambda_j(k+1)\mathbf{1}_n - \gamma_{ji}(k+1) + \gamma_{ji}(k+1)\eta_{i}(k) \\
				& = 1 - \gamma_{ji}(k+1)\left(1-\eta_{i}(k)\right) \triangleq \eta_{j} < 1
			\end{split}
		\end{equation*}
		Swithching $j$ to $i$, and combine the two cases above, we have
		\begin{equation*}
			\begin{split}
				\left\lVert \boldsymbol{\bar{p}}_{i}(k+2)\right\rVert &\leq 
				\Lambda_i(k+1)\Lambda(k)\mathcal{P}(k) + \delta \sum_{t=k}^{k+1}\lambda^{t} \\
				&\leq \Lambda_i(k+1)\Lambda(k)\mathbf{1}_n\mathsf{Q}(k) + \delta \sum_{t=k}^{k+1}\lambda^{t} \\
				&\leq \eta_{i} \mathsf{Q}(k) + \delta \sum_{t=k}^{k+1}\lambda^{t} , \forall i\in \mathcal{V}
			\end{split}
		\end{equation*}
		Further, by letting $\eta=\min_{i\in\mathcal{V}}\eta_{i} \in (0,1) $, we can obtain
		\begin{equation}\label{eq:pibarconv}
			\mathsf{Q}(k+2) \leq \eta \mathsf{Q}(k) + \delta \sum_{t=k}^{k+1}\lambda^{t}
		\end{equation}
		which implies the exponentially convergent of $\boldsymbol{\bar{p}}_{i}$, i.e. the center of the time-varying formation $\boldsymbol{p}(k)$ converges to the target's position $\boldsymbol{p}_0(k)$ exponentially fast under the tracking control law \eqref{eq:formationtrackingcontrol}.

	\end{proof}

	\begin{figure}[tbp]
		\centering
		\includegraphics[width=\linewidth]{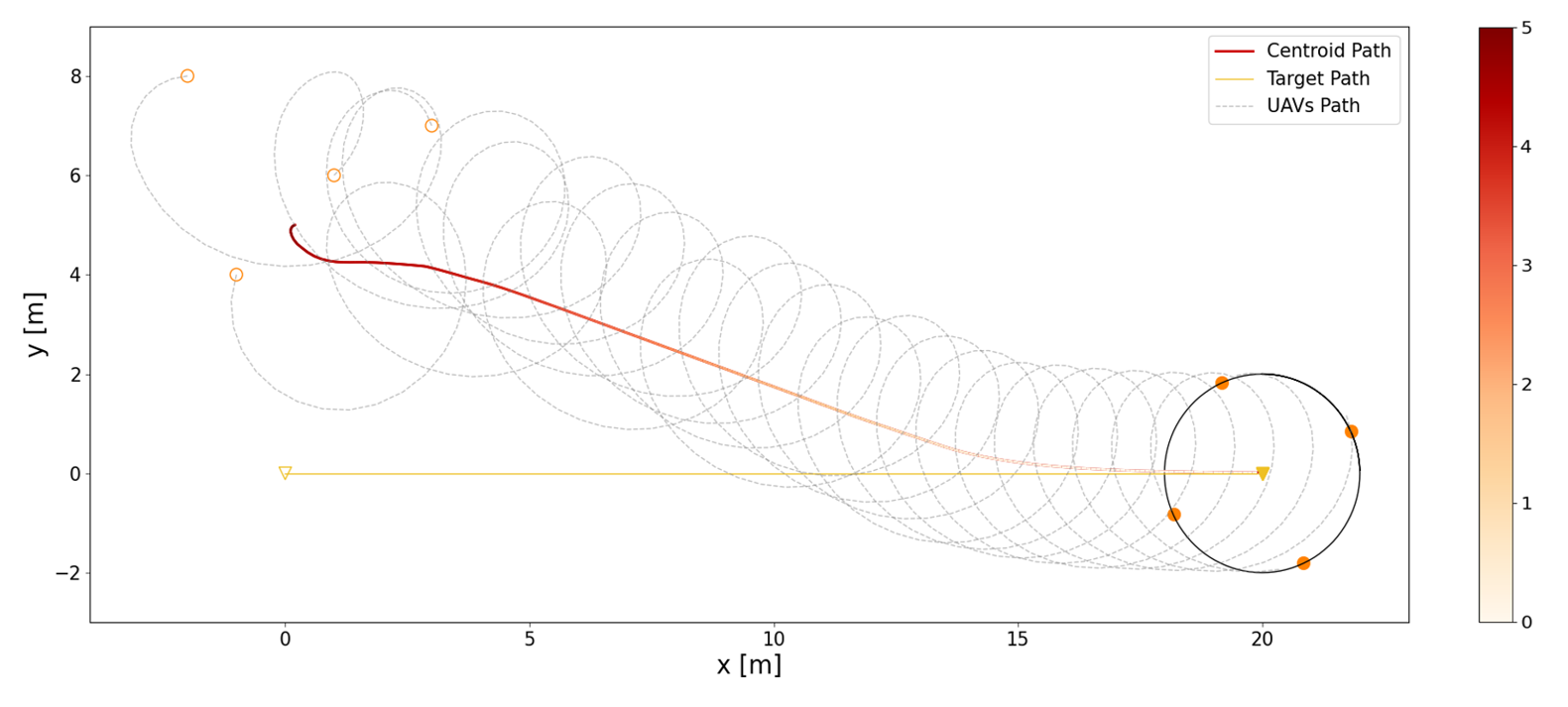}
		\caption{ Four UAVs tracking a moving target. The centroid path gradually converges to the target path. The gradient color of the path indicates a gradual reduction in tracking error.   }
		\label{fig:tracking}
	\end{figure}

	\begin{figure}[tbp]
		\centering
		\includegraphics[width=\linewidth]{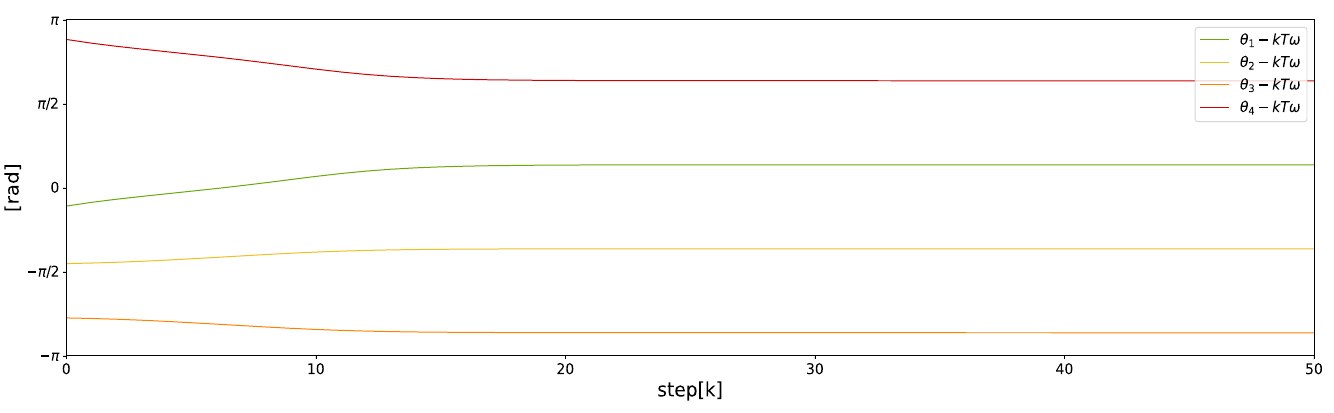}
		\caption{The relative phase of UAVs, denoted by $\theta_i(k)-kT\omega_i $.   }
		\label{fig:theta}
	\end{figure}

	\begin{figure}[tbp]
		\centering
		\includegraphics[width=0.95\linewidth]{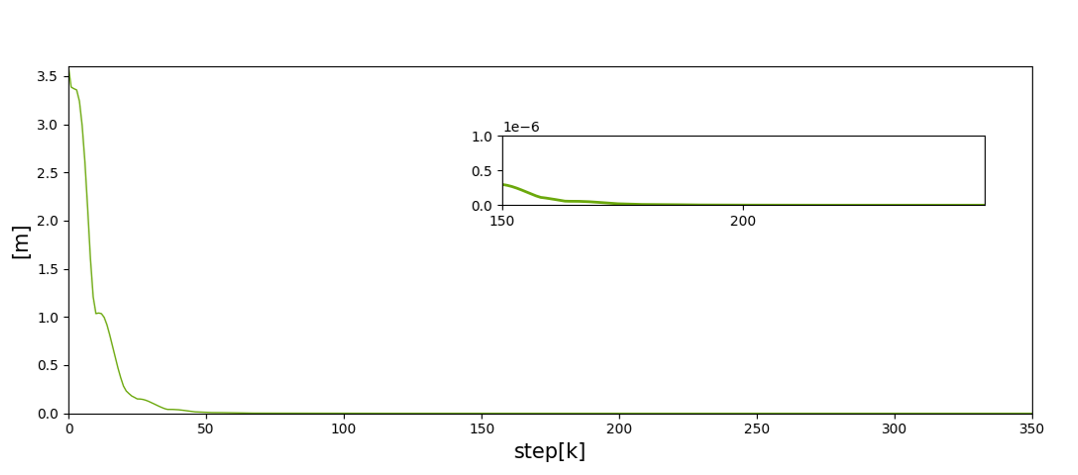}
		\caption{The relative position error $\max_{(i,j)\in\mathcal{\bar{E}}} \left\lVert \boldsymbol{\tilde{p}}_{ij}(k) \right\rVert$.  }.
		\label{fig:rlerror}
	\end{figure}

	\begin{figure}[tbp]
		\centering
		\includegraphics[scale=0.3]{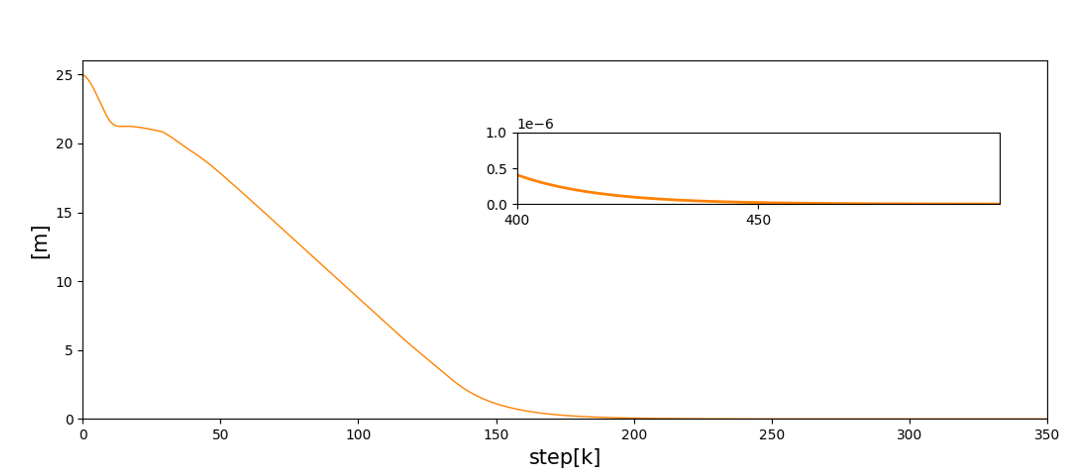}
		\caption{The formation tracking error $\left\lVert \boldsymbol{\bar{p}}(k) \right\rVert$.  }
		\label{fig:trackerror}
	\end{figure}

	\section{Simulation and Experiments} \label{sec:simulation}

    \subsection{Numerical Simulations}
	This section contains simulation results to support the theoretical findings in this paper. We consider four UAVs tracking a moving target that only UAV $1$ senses the target. For the discrete time Kuramoto model based algorithm \eqref{eq:osc}, we give different initial value for $ \theta_i(0) $, set the frequency $\omega_i = \pi / 2 $ , and choose $K_l = 1$ for $ l\in\{1,...,n-1\}$, and $K_n = -1$. Let $\boldsymbol{\hat{p}}_{ij}(0) = \mathbf{0}_m $ and $ \varGamma_{ij}(0) = \mathbf{I} $ as the initial value for the least-square estimator \eqref{eq:estimator}, and choose $\beta_f=0.7$. We select $T=0.125$, $\beta=7$, $\bar{U}_{i}=0.4$ , and $\bar{a}_{ij}=1/\|\mathcal{N}_i\| $ for the formation tracking controller \eqref{eq:formationtrackingcontrol}. 
	
	The first simulation result, Fig.~\ref{fig:shapedy}, shows the ability to cope with complex environments through affine transformation. The time-varying reference of the radius is $\rho = 2\sin(k\pi /100)+4 $ by setting $S_x=S_y=0.5\sin(k\pi /100)+1 $ and the center corresponds to the moving target with $ \boldsymbol{p}_0(k) =\left[ 0.1k,3\sin(\pi/150*k) \right]^{T}  $. At time step $k=250$, we assume that UAV $4$ is damaged, and the remaining UAVs still can circle the target in a balanced pattern.

	As shown in Fig.~\ref{fig:tracking}, we simulated tracking a moving target with $ \boldsymbol{p}_0(k) =\left[ 0.1k,0 \right]^{T} $. The centroid path of the formation gradually approaches the target path. Fig.~\ref{fig:theta} shows the relative phase between UAVs, thus pattern formation is achieved by the oscillator network \eqref{eq:osc}. The relative position error and the formation tracking error can be seen in Fig.~\ref{fig:rlerror} and Fig.~\ref{fig:trackerror} respectively.

	\begin{figure}[tbp]
		\centering
		\includegraphics[width=\linewidth]{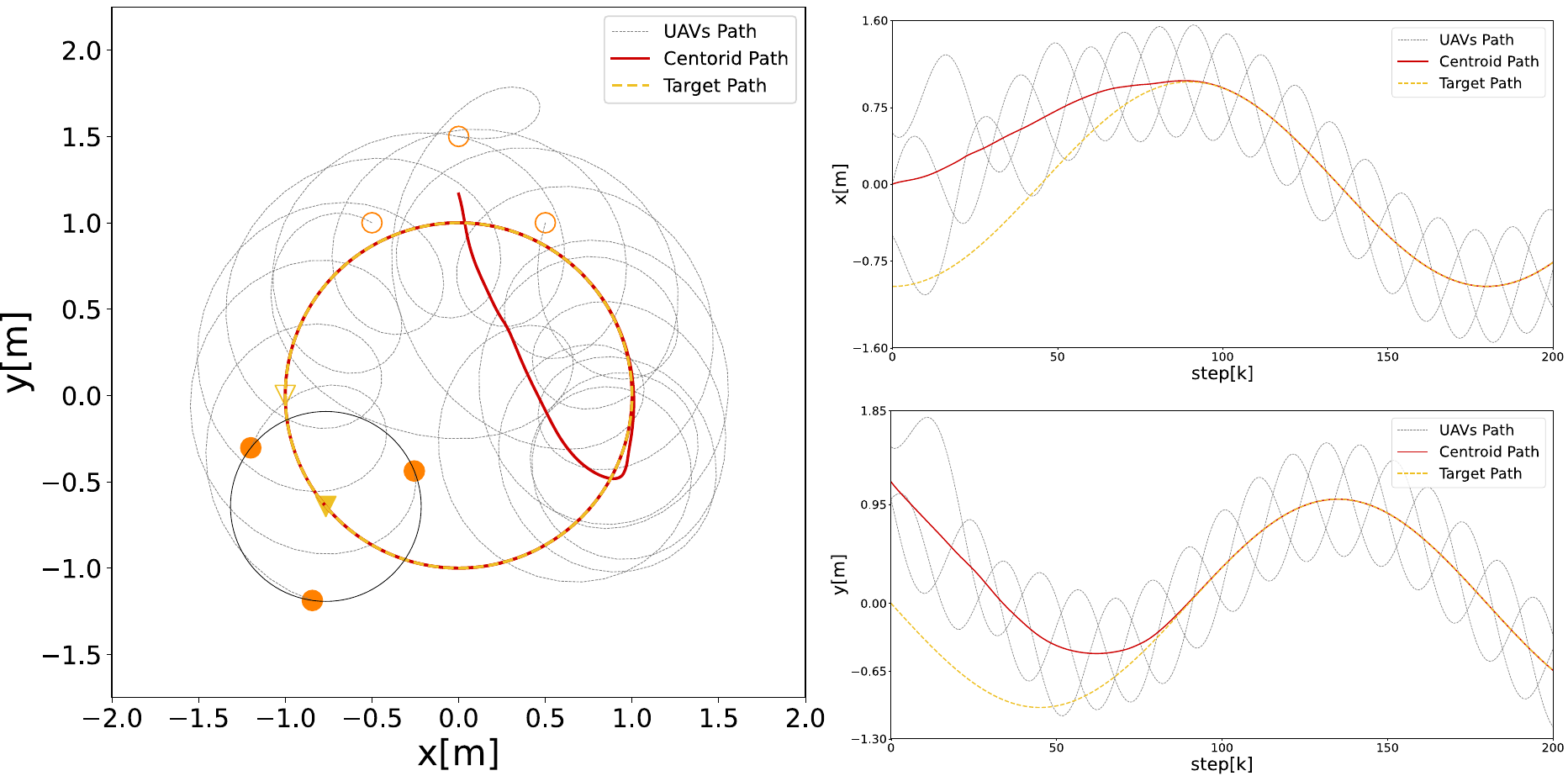}
		\caption{ Experimental results: (a) The path of the UAVs in $XY$-plane. (b) The position of UAVs in $x$-axis and $y$-axis respectively. The initial positions are represented in void style, and given the position of UAVs at time step $k=250$ in solid style.}.
		\label{fig:realxy}
	\end{figure}

	\begin{figure}[tbp]
		\centering
		\includegraphics[scale=0.3]{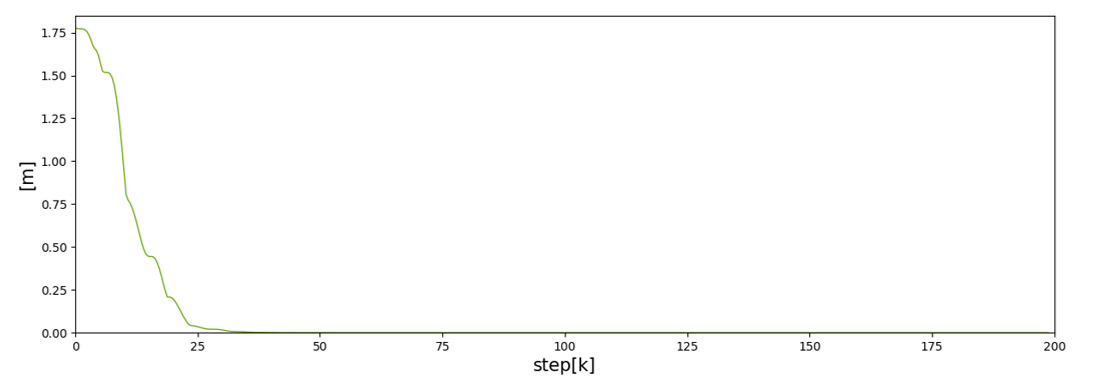}
		\caption{ Experimental results: the relative position error $\max_{(i,j)\in\mathcal{\bar{E}}} \left\lVert \boldsymbol{\tilde{p}}_{ij}(k) \right\rVert$.  }
		\label{fig:realrl}
	\end{figure}

	\begin{figure}[tbp]
		\centering
		\includegraphics[scale=0.23]{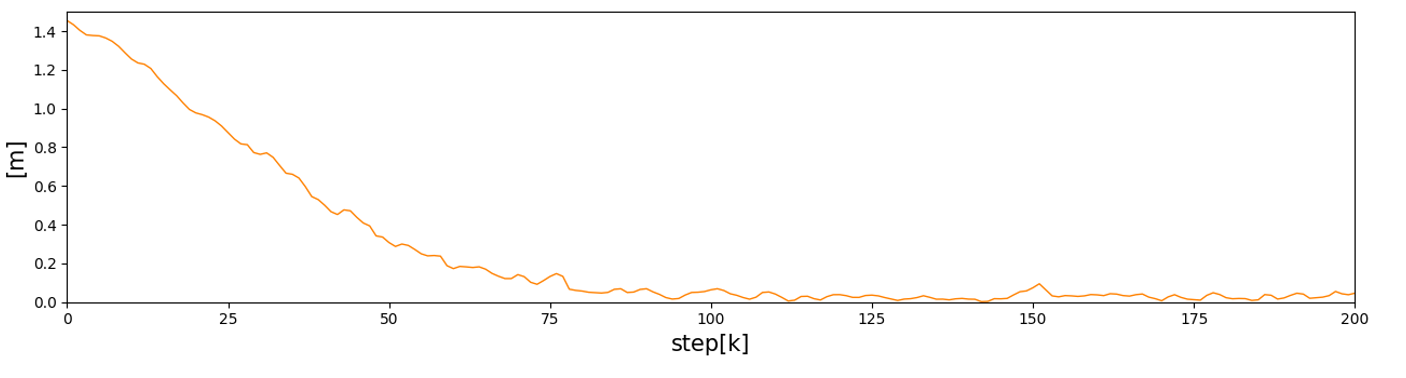}
		\caption{Formation tracking error $\left\lVert \boldsymbol{\bar{p}}(k) \right\rVert$ in the experiment.  }
		\label{fig:realtrack}
	\end{figure}
 
    \subsection{Experimental Results }
	To further verify the real-time performance of the proposed method, an experiment of four Crazyflie 2.0 UAVs is implemented in an indoor testing area. One of the Crazyflies was considered as the target, flying at an altitude of $z=0.3m$ above the ground in a circle trajectory $ \boldsymbol{p}_0(k) =\left[ \cos(k\pi/90),\sin(k\pi/90) \right]^{T}  $. The other three UAVs will be tracking the target at a higher altitude $z=1.4m$. The UAVs are tracked by the OptiTrack motion capture system. We choose $T=0.1$, $\bar{U}_{i}=0.3$, and assume that all the UAVs can sense the target one. Thanks to the cascaded PID controller of Crazyflie 2.0, we control the UAVs by giving desired setpoints \eqref{eq:dynamics}. The trajectory of the UAVs is mapped to $XY$-plane (see Fig.~\ref{fig:realxy}). Experimental results of the relative position error and the formation tracking error can be seen in Fig.~\ref{fig:realrl} and Fig.~\ref{fig:realtrack} respectively. The errors may be due to the instability caused by wind disturbances between the UAVs and the existence of measurement errors. From the experimental results, we can conclude that the proposed control algorithm reaches a reasonable performance. 

    \section{Conclusion} \label{sec:conclution}
    This paper presented a new control strategy to stabilize a group of UAVs to enclose a moving target via a distance-displacement-based relative localization method. By embedding the Kuramoto-based pattern design scheme into RLSE, we met the persistent excitation condition of relative displacements which is further facilitate the exponential convergence of relative position estimation. Then, a new framework based on affine transformations was proposed to obtain more complex time-varying formations to respond to more practical needs. Based on the relative position estimate, a formation tracking control scheme was designed to achieve the center of the formation tracking the moving target. The main advantage of this approach is that in the case of only a small number of UAVs sensing the target, the enclosing objective can also be realized. Numerical and empirical examples were used to demonstrate the theoretical findings.
    Future works will focus on extending this control strategy including some additional constraints such as adding potential terms to ensure obstacle avoidance and planning more flexible trajectories while ensuring that the persistent excitation condition is satisfied. 
        

	\bibliography{CDC2023.bib,references.bib}
	\bibliographystyle{IEEEtran}
	
\end{document}